\documentclass[11pt,a4paper]{report}

\usepackage[utf8]{inputenc}
\usepackage{fullpage}
\usepackage{fp}
\usepackage{color}
\usepackage{url}
\usepackage{graphicx}
\usepackage{fancyhdr}
\usepackage{framed}
\usepackage{lastpage}
\usepackage{listings}
\usepackage{alltt}
\usepackage{multicol}

\usepackage{listings}
\usepackage{alltt}
\usepackage[all]{xy}
\usepackage{algorithmic}
\usepackage[colorlinks,urlcolor=blue,linkcolor=magenta,citecolor=red,linktocpage=true]{hyperref}
\usepackage{subfig}
\usepackage{amsmath, amsthm, amssymb}

\newcommand{\PKG}[1]{\texttt{#1}}
\newcommand{\FOSS}{FOSS}
\newcommand{\MANCOOSI}{Mancoosi}

\newcommand{\DEP}{\ensuremath{\rightarrow}}	
\newcommand{\SDEP}{\ensuremath{\Rightarrow}}	
\newcommand{\SPREDS}[1]{\ensuremath{Spreds(#1)}}	
\newcommand{\SCONS}[1]{\ensuremath{Scons(#1)}}	

\newcommand{\DSENS}[1]{\ensuremath{|#1|}}
\newcommand{\SENS}[1]{\ensuremath{||#1||}}
\newcommand{\IS}[2]{\ensuremath{Is(#1,#2)}}
\newcommand{\DIS}[2]{\ensuremath{DirIs(#1,#2)}}

\newcommand{\SDOM}[2]{\ensuremath{#1\succcurlyeq_{Is}#2}}
\newcommand{\SDOMUPTO}[3]{\ensuremath{{#1}\succcurlyeq_{Is}^{#3}{#2}}}

\newcommand{\PRETHM}{}
\newcommand{\POSTTHM}{}

\newcounter{myindex}
\newcommand{\INCR}{\stepcounter{myindex}\arabic{myindex}}

\lstdefinestyle{debctrl}
               {basicstyle=\small\normalfont\ttfamily,
                 showstringspaces=false,
                 emph={Package,Version,Build,Depends,Conflicts,Provides},
                 keywordstyle=\color{black}\bf,
                 emphstyle=\bf,
               }

\newtheorem{definition}{Definition}[section]

\newtheorem{proposition}[definition]{Proposition}
\newtheorem{example}[definition]{Example}
\newtheorem{remark}[definition]{Remark}

\advance\topmargin by -3mm

\definecolor{mancoosi@lightblue}{RGB}{1,94,140}
\definecolor{mancoosi@green}{RGB}{157,199,218}
\definecolor{mancoosi@red}{RGB}{189,16,11}

\title{Technical Report : Strong Dependencies between Software
  Components \footnote{Partially supported by the European Community's
    7th Framework Programme (FP7/2007-2013), grant agreement
    n${}^\circ$214898, ``\MANCOOSI'' project.}}

\renewcommand{\thesection}{\arabic{section}} 

\begin{document}
\begin{titlepage}
%
%
\begin{minipage}[p]{7cm}
\begin{large}
\textcolor{mancoosi@lightblue}{
\begin{tabular}{l}
Specific Targeted Research Project\\
Contract no.214898\\
Seventh Framework Programme: FP7-ICT-2007-1\\ 
\end{tabular}
}
\end{large}
\end{minipage}
\hfill
\includegraphics*[width=4cm]{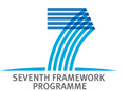} 

%
%

\includegraphics*[width=5cm]{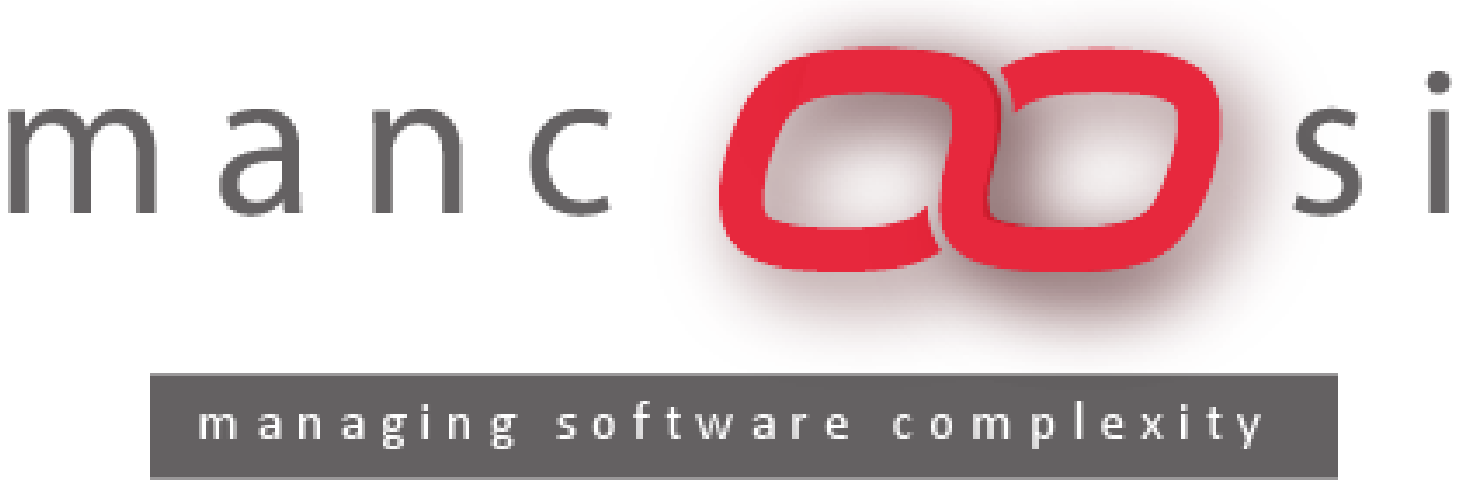} \hfill
\textcolor{mancoosi@green}{\LARGE\bf Technical Report 0002}

\vfill

\begin{center}

%
%
{\bf \Large MANCOOSI} \\[3em]

{\bf \Large Managing the Complexity of the Open Source Infrastructure}\\
\end{center}

\vfill 
%
%
\begin{framed}
\centering{
  \textcolor{mancoosi@lightblue}{\Huge \bf Strong Dependencies between
    Software Components}}
\end{framed}
\vspace{1cm}
\begin{center}\Large
  \begin{tabular}{c@{\hspace{1cm}}c}
    Pietro Abate (pietro.abate@pps.jussieu.fr) \\ 
    Jaap Boender (jaap.boender@pps.jussieu.fr)\\
    Roberto Di Cosmo (roberto@dicosmo.org) \\
    Stefano Zacchiroli (zack@pps.jussieu.fr)
  \end{tabular}\\[2ex]
  Universitè Paris Diderot, PPS\\
  UMR 7126, Paris, France
\end{center}

\vfill

\begin{center}
\large \today
\end{center}

\vfill

\begin{center}
  \bf\large Web site: \url{www.mancoosi.org}
\end{center}

\end{titlepage}

\tableofcontents
\section*{Abstract}

Component-based systems often describe context requirements in terms
of explicit inter-component dependencies. Studying large instances of
such systems---such as free and open source software (\FOSS)
distributions---in terms of declared dependencies between packages is
appealing. It is however also misleading when the language to express
dependencies is as expressive as boolean formulae, which is often the
case. In such settings, a more appropriate notion of component
dependency exists: \emph{strong dependency}. This paper introduces
such notion as a first step towards modeling semantic, rather then
syntactic, inter-component relationships.

Furthermore, a notion of component \emph{sensitivity} is derived from
strong dependencies, with applications to quality assurance and to the
evaluation of upgrade risks. An empirical study of strong dependencies
and sensitivity is presented, in the context of one of the largest,
freely available, component-based system.

\newpage

\section{Introduction}
\label{sec:intro}

Component-based software architectures~\cite{szyperski} have the
property of being upgradeable piece-wise, without necessarily touching
all the pieces at the same time. The more pieces are affected by a
single upgrade, the higher the impact of the upgrade can be on the
usual operations performed by the overall system; this impact can
either be beneficial (if the upgrade works as planned) or disastrous
(if not). Package-based \FOSS{} (Free and Open Source Software)
distributions are possibly the largest-scale examples of
component-based architectures, their upgrade effects are experienced
daily by million of users world-wide, and the historical data
concerning their evolution is publicly available.

Within \FOSS{} distributions, software components are managed as
\emph{packages}~\cite{hotswup:package-upgrades}. Packages are
described with meta-information, which include complex
inter-relationships describing the static requirements to run properly
on a target system. Requirements are expressed in terms of other
packages, possibly with restrictions on the desired versions. Both
positive requirements (\emph{dependencies}) and negative requirements
(\emph{conflicts}) are usually allowed.

\PRETHM
\begin{example}
\label{ex:postfix}
An excerpt of the inter-package relationships of the \PKG{postfix}
Internet mail transport agent in Debian
GNU/Linux\footnote{\url{http://www.debian.org}} currently reads:
\begin{lstlisting}[style=debctrl]
Package: postfix
Version: 2.5.5-1.1
Depends: libc6 (>= 2.7-1), libdb4.6, ssl-cert,
  libsasl2-2, libssl0.9.8 (>= 0.9.8f-5),
  debconf (>= 0.5) | debconf-2.0,
  netbase, adduser (>= 3.48), dpkg (>= 1.8.3),
  lsb-base (>= 3.0-6)
Conflicts: libnss-db (<< 2.2-3), smail,
  mail-transport-agent, postfix-tls
Provides: mail-transport-agent, postfix-tls
\end{lstlisting}
\end{example}
\POSTTHM

As this short example shows, inter-package relationships can get quite
complex, and there are plenty of more complex examples to be found in
distributions like Debian. In particular, the language to express
package relationships is not as simple as \emph{flat} lists of
component predicates, but rather a structured language whose syntax
and semantics is expressed by conjunctive normal form (CNF)
formulae~\cite{EdosAse06}. In Example~\ref{ex:postfix}, commas
represent logical conjunctions among predicates, whereas bars
(``\texttt{|}'') represent logical disjunctions. Also, indirections by
the mean of so-called \emph{virtual packages} can be used to declare
feature names over which other packages can declare relationships; in
the example (see: ``Provides'') the package declares to
provide the features called \PKG{postfix-tls} and
\PKG{mail-transport-agent}.

Within this setting, it is interesting to analyse the \emph{dependency
  graph} of all packages shipped by a mainstream \FOSS{}
distribution. This graph is potentially very large as distributions
like Debian are composed of several tens of thousands packages but it
is surely smaller than widely studied graphs such as the World Wide
Web graph~\cite{albert_www_diameter}. It is also more expressive, in
the sense that it contains different types of edges (dependencies and
conflicts for example) and allows the use of disjunctions to express alternative
paths.  Simple graph encodings of the
package universe have been proposed in the
past~\cite{labelle_os_networks,maillart_empirical_2008}, to study the
adherence of the dependency graph to small-world network laws. In such
encodings, inter-package relationships were approximated by a simple
binary relation of \emph{direct} dependency, which is noted $p\DEP q$
in this paper. Formally, $p\DEP q$ holds whenever package $q$ occurs
syntactically in the dependency formula of $p$. This notion of direct
dependency does not distinguish between $q$ occurring in conjunctive
or disjunctive position, ignoring the semantic difference between
conjunctive and disjunctive dependencies, as well as the presence of
conflicts among components.

In this paper we argue that there is a different dependency graph to
be studied to grasp meaningful relationships among software
components: a graph that represents the \emph{semantics} of
inter-component relationships, in which an edge between two components
is drawn only if the first cannot be installed without installing the
second. We call such a graph the \emph{strong dependency graph}, argue
that it is better suited to study package universes in component-based
architectures, and study its network properties.  Finally, we argue
that the strong dependency graph can be used to establish a measure of
package ``sensitivity'' which has several uses, from distribution wide
quality assurance to establishing the potential risks of package
upgrades. As a relevant, yet empirical, case study we build and
analyse the strong dependency graph of present and past \FOSS{}
distributions, as well as the corresponding package sensitivity.

The rest of the paper is structured as follows:
Section~\ref{sec:strongdeps} introduces the notion of strong
dependency, highlights the differences with plain dependencies and
proposes related sensitivity metrics. Section~\ref{sec:measures}
computes dependencies and sensitivity of components of a large and
popular \FOSS{} distribution. Section~\ref{sec:algo} gives an
efficient algorithm to compute strong dependencies for large software
repositories. Section~\ref{sec:applications} discusses applications of
the proposed metrics for quality assurance and upgrade risk
evaluation. Before concluding, Section~\ref{sec:related} discusses
related research.

\section{Strong dependencies}
\label{sec:strongdeps}

Component dependencies can be used to compute relevant quality
measures of software repositories, for instance to identify
particularly fragile
components~\cite{dick-mining-db,mining-sw-quality,dynamine}.  It is
well known that small-world networks are resilient to random failures
but particularly weak in the presence of attacks, due to the existence
of highly connected \emph{hub nodes} \cite{albert-2000-406}. To
identify the components whose modification (e.g., removal or upgrade)
can have a high potential impact on the stability of a complex
software system, it is natural to look for \emph{hubs} on which a lot
of other components depend.

In \FOSS{} distributions, as well as many other component-based
systems~\cite{apache-maven,clayberg-eclipse-plugin}, the language used
to express inter-package relationships is expressive enough to cover
propositional logic. As a consequence, considering only \emph{plain
  connectivity}---i.e., the possibility of going from one package to
another following dependency arcs---is no longer meaningful to
identify hubs. For example, if $p$ is to be installed and there exists
a dependency path from $p$ to $q$, it is not true that $q$ is always
needed for $p$, and in some cases $q$ may even be incompatible with
$p$.

In other terms, the \emph{syntactic} connectivity notion does not tell
much about the real structure of dependencies: we need to go further
and analyse the \emph{semantic} connectivity among software components
induced by the explicit dependencies in the graph. That led us to the
following definition.

\PRETHM
\begin{definition}[Strong dependency]
  \label{def:strongdep}
  Given a repository $R$, we say that a package $p$ in $R$
  \emph{strongly depends} on a package $q$ in $R$, written $p\SDEP_R
  q$, if there exists a healthy installation of $R$ containing $p$,
  and every healthy installation of $R$ containing $p$ also contains
  $q$.  We write $\SPREDS{p}_R$ for the set $\{q|q\SDEP_R p\}$ of
  strong predecessors of a package $p$ in $R$, and $\SCONS{p}_R$ for
  the set $\{q|p\SDEP_R q\}$ of strong successors of $p$ in $R$.
\end{definition}
\POSTTHM

In the following, we will drop the $R$ subscript when the repository
is clear from the context.

The notions of repository and healthy installation used here are
from~\cite{EdosAse06}; the underlying intuitions are as follows. A
\emph{repository} is a set of packages, together with dependencies and
conflicts encoded as propositional logic predicates over other
packages contained therein; an \emph{installation} is a subset of the
repository; an installation is said to be \emph{healthy} when all its
packages have their dependencies satisfied within the installation and
dually their conflicts \emph{un}satisfied.

Intuitively, $p$ strongly depends on $q$ with respect to $R$ if it is
not possible to install $p$ without also installing $q$. Notice that
the definition requires $p$ to be installable in $R$ as otherwise it
would vacuously depend on all the packages $q$ in the repository.  Due
to the complex nature of dependencies, there can be a huge gap with
the syntactic dependency graph as naively extracted from the metadata.

\PRETHM
\begin{example}[Direct and strong dependencies]
  In simple cases, conjunctive direct dependencies translate to
  identical strong dependencies whereas disjunctive ones vanish, as
  for the packages of the following repository:
  \begin{multicols}{3}
    \begin{lstlisting}[style=debctrl]
Package: p
Depends: q, r

Package: a
Depends: b | c
    \end{lstlisting}
    \columnbreak
    \xymatrix@C=1pc{
      & p\ar[dl]\ar[dr] &\\
      q & & r 
    }
    \columnbreak
    \xymatrix@C=1pc{
      & a\ar^<<<{\bigvee}[dl]\ar[dr] &\\
      b & & c 
    }
  \end{multicols}

  We have that $p\DEP q, p\DEP r$ and $p\SDEP q, p\SDEP r$ (because
  $p$ cannot be installed without either $q$ or $r$), and that $a\DEP
  b, a\DEP c$ whereas $a\not\SDEP b, a\not\SDEP c$ (because $a$ does not
  forcibly require neither $b$ nor $c$).  In general however, the
  situation is much more complex, like in the following repository:

  \begin{multicols}{2}
    \begin{lstlisting}[style=debctrl]
Package: p
Depends: q | r

Package: r
Conflicts: p

Package: q
    \end{lstlisting}
    \columnbreak
    \xymatrix@C=1pc{
      & p\ar^<<<{\bigvee}[dl]\ar[dr] &\\
      q & & r\ar@/_/@{.}[ul]_\# 
    }
  \end{multicols}
  Notice that $p\SDEP q$ in spite of $q$ not being a conjunctive
  dependency of $p$, and $r$ is incompatible with $p$, despite the
  fact that $p\DEP r$.
\end{example}
\POSTTHM

\PRETHM
\begin{proposition}[Transitivity]
  If $p\SDEP_R q$ and $q\SDEP_R r$ then $p\SDEP_R r$.
\end{proposition}
\begin{proof}
  Trivial from Definition~\ref{def:strongdep}.
\end{proof}
\POSTTHM

On top of the strong and direct dependency notions, we can define the
corresponding \emph{dependency graphs}.

\PRETHM
\begin{definition}[Dependency graphs]
  The \emph{strong dependency graph} $SG(R)$ of a repository $R$ is
  the directed graph having as vertices the packages in $R$ and as
  edges all pairs $\langle p,q\rangle$ such that $p\SDEP q$. Note that
  the $SG(R)$ is transitively closed as direct consequence as the
  transitivity of the strong dependency relation.

  Similarly, the \emph{direct dependency graph} $DG(R)$ is the
  directed graph having as vertices the packages in $R$ and as edges
  all pairs $\langle p,q\rangle$ such that $p\DEP q$.
\end{definition}
\POSTTHM

The dependency graphs can be used to formalise, via the notion of
\emph{impact set}, the intuitive notion of the set of packages which
are potentially affected by changes in a given package.
\PRETHM
\begin{definition}[Impact set of a component]
  Given a repository $R$ and a package $p$ in $R$, the \emph{impact
  set} of $p$ in $R$ is the set $\IS{p}{R}=\{q\in R ~|~ q \SDEP p \}$.

  Similarly, the \emph{direct impact set} of $p$ is the set
  $\DIS{p}{R}=\{q\in R ~|~ q \DEP p \}$.
\end{definition}
\POSTTHM

While the impact set gives a sound lower bound to the set of packages
which can be potentially affected by a change in a package, the direct
impact set offers no similar guarantees. Note that by
Definition~\ref{def:strongdep}, for all package $p$, $p\in\IS{p}{R}$.
Package sensitivity---a measure of how sensitive is a package, in
terms of how many other packages can be affected by a change in
it---can now be defined as follows.
\PRETHM
\begin{definition}[Sensitivity]
  The strong sensitivity, or simply \emph{sensitivity}, of a package
  $p\in R$ is $|\IS{p}{R}|-1$, i.e., the cardinality of the impact set
  minus 1.\footnote{The $-1$ accounts for the fact that the impact set
    of a package always contains itself. This way we ensure that
    sensitivity 0 preserves the intuitive meaning of ``no package
    potentially affected''.}

  Similarly, the \emph{direct sensitivity} is the cardinality of the
  direct impact set.
\end{definition}
\POSTTHM

The higher the sensitivity of a package $p$, the higher the
\emph{minimum} number of packages which will be potentially affected
by a change, such as a new bug, introduced in $p$. We write \DSENS{p}
and \SENS{p} to denote the direct and strong sensitivity of package
$p$, respectively. The following basic property of impact sets and
sensitivity follows easily from the definitions.

\PRETHM
\begin{proposition}[Inclusion of impact sets]
  If $p\SDEP_R q$ then $\IS{p}{R}\subseteq \IS{q}{R}$. As a
  consequence, the sensitivity of $p$ in $R$ is smaller than the
  sensitivity of $q$ in $R$.
\end{proposition}
\POSTTHM

When analysing a large component base, like Debian's, which contains
about 22'000 components, it is important to be able to identify some
measure that can be used to easily pinpoint ``interesting'' packages.
Sensitivity can be (and actually is, in our tools) used to order
packages, bringing the most sensitive to the forefront. But
sensitivity alone is not enough: we do not want to spend time going
through hundreds of packages with similar sensitivity to find the one
which is really important, so we need to keep some of the structure of
the strong dependency graph.

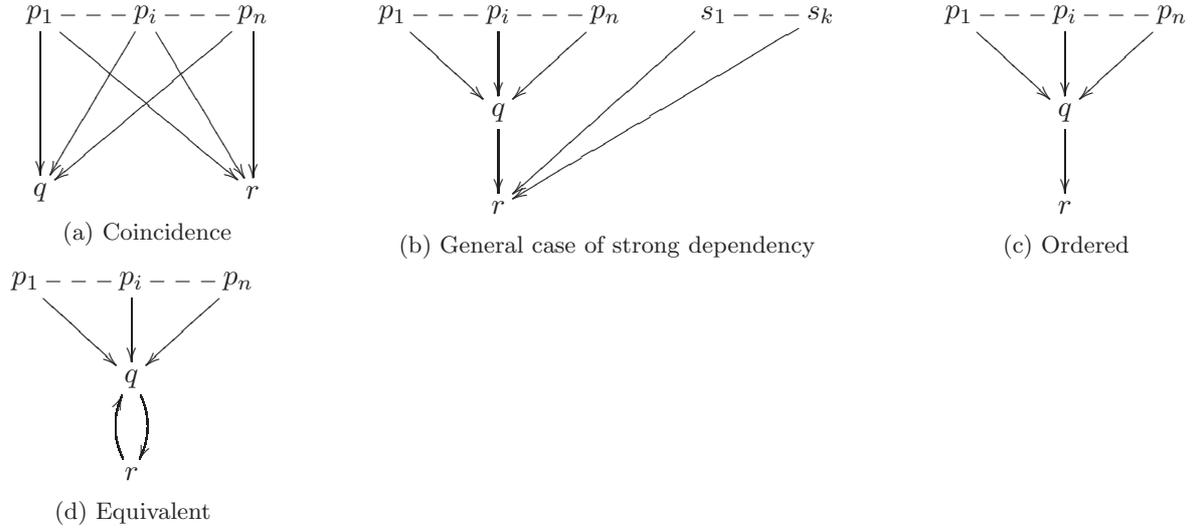
\begin{figure*}[t!]
  \hspace{-0.8cm}
\subfloat[Coincidence]{
\xymatrix{
p_1\ar[dd]\ar[ddrr]\ar@{--}[r] & p_i\ar[ddr]\ar[ddl]\ar@{--}[r]& p_n\ar[ddll]\ar[dd]\\
\\
q    &  & r\\
}
\label{f:unrelated}
}
\subfloat[General case of strong dependency]{
 \xymatrix{
 p_1\ar[dr]\ar@{--}[r] & p_i\ar[d]\ar@{--}[r]& p_n\ar[dl]&s_1\ar[ddll]\ar@{--}[r] & s_k\ar[ddlll]\\
     & q\ar[d] & \\
     & r &
 }
\label{f:strongdep}
}
\subfloat[Ordered]{
 \xymatrix{
 p_1\ar[dr]\ar@{--}[r] & p_i\ar[d]\ar@{--}[r]& p_n\ar[dl]\\
     & q\ar[d] & \\
     & r &
 }
\label{f:ordered}
}
\subfloat[Equivalent]{
 \xymatrix{
 p_1\ar[dr]\ar@{--}[r] & p_i\ar[d]\ar@{--}[r]& p_n\ar[dl]\\
     & q\ar@/^/[d] & \\
     & r\ar@/^/[u] &
 }
\label{f:equivalent}
}
\caption{Some significant configurations in the strong dependency
  graph}\label{f:peculiarcases}
\end{figure*}

A first step is to group together only those packages that are related
by strong dependencies, but our analysis of the Debian distribution
led us to discover that we really need to go further and distinguish
the cases of related components in the strong dependency graph from
the cases of unrelated ones: in the picture in
figure~\ref{f:peculiarcases}\footnote{Edges implied by transitivity
  are omitted from the diagrams for the sake of clarity.},
configuration~\ref{f:ordered} shows $q$ that clearly dominates $r$, as
the impact set of $r$ really comes from that of $q$, in
configuration~\ref{f:equivalent}, $q$ and $r$ are clearly equivalent,
while in configuration~\ref{f:unrelated}, $q$ and $r$ are totally
unrelated, and in configuration~\ref{f:strongdep}, $q$ strong depends
on $r$ but $q$ does not generate all the impact set of $r$.

Yet, the packages $q$ and $r$ all have essentially the same
sensitivity values ($n$ or $n+1$) in all the first three cases (and
$n+k$ in the fourth, which can also contribute to the mass of packages
of sensitivity similar to $n$).  To distinguish these different
configurations in strong dependency graphs, we introduce one last
notion.

\PRETHM
\begin{definition}[Strong dominance]\label{def:dominance}
Given two packages $p$ and $q$ in a repository $R$, we say that $p$
strongly dominates $q$ (\SDOM{p}{q}) iff
\begin{itemize}
\item $\IS{p}{R} \supseteq (\IS{q}{R}\setminus \SCONS{p})$, and
\item $p$ strongly depends on $q$
\end{itemize}
\end{definition}
\POSTTHM

The intuition of strong dominance, is that a package $p$ dominates
$q$ if the strong dependency of $p$ on $q$ \emph{explains} the impact
set of $q$: the packages that $q$ has an impact on are really those
that $p$ has an impact on, plus $p$. This notion has some similarity
in spirit with the standard notion of dominance used in control flow
graphs, but is technically quite different, as strong dependency 
graphs are transitive, and have no start node.

Using the transitivity of strong dependencies, the following can be
established.

\begin{proposition}
  The strong domination relation is a partial pre-order.
\end{proposition}
\begin{proof}
  Reflexivity is trivial to check.  For transitivity, suppose we have
  \SDOM{p}{q} and \SDOM{q}{r}: first of all, $p$ strongly depends on
  $r$ is a direct consequence of the fact that the strong dependency
  relation is transitive, so the second condition for \SDOM{p}{r} is
  established. For the first condition, we know that $\IS{p}{R}
  \supseteq ( \IS{q}{R}\setminus \SCONS{p})$ and $\IS{q}{R}\supseteq
  (\IS{r}{R}\setminus \SCONS{q})$. By transitivity of strong
  dependencies, since $p\SDEP q\SDEP r$, we also have that
  $\SCONS{p}\supseteq\SCONS{q}\supseteq\SCONS{r}$.  Then we have
  easily that $\IS{p}{R}\supseteq(\IS{q}{R}\setminus \SCONS{p})
  \supseteq (\IS{r}{R}\setminus \SCONS{q})\setminus \SCONS{p} =
  \IS{r}{R}\setminus \SCONS{p}$.
\end{proof}
\POSTTHM

This pre-order is now able to distinguish among the cases of
Figure~\ref{f:peculiarcases}. In Figure~\ref{f:ordered} we have that
$\SDOM{q}{r}$, but not the converse; in~\ref{f:equivalent} both
$\SDOM{q}{r}$ and $\SDOM{r}{q}$ hold, i.e., $q$ and $r$ are equivalent
according to strong domination; in~\ref{f:unrelated}
and~\ref{f:strongdep} no dominance relationship can be established
between $q$ and $r$.

It is possible, and actually quite useful, to generalise the strong
dominance relation to cover also the case shown in~\ref{f:strongdep},
where a part of the impact set of the package $r$ is not covered by
the impact set of $q$, as follows.

\PRETHM
\begin{definition}[Relative strong dominance]\label{def:reldominance}\mbox{}\\
Given two packages $p$ and $q$ in a repository $R$, we say that $p$
strongly dominates $q$ up to $z$ (\SDOMUPTO{p}{q}{z}) iff
\begin{itemize}
\item ${{|(\IS{q}{R}\setminus \SCONS{p}) \setminus \IS{p}{R}|}\over{|\IS{p}{R}|}}*100 = z $, and
\item $p$ strongly depends on $q$
\end{itemize}
\end{definition}
\POSTTHM

It is easy to see that $\SDOM{p}{q}$ iff $\SDOMUPTO{p}{q}{0}$, and one
can compute in a single pass on the repository the values $z$ for each
pair of packages such that $p\SDEP q$, leaving for later the choice of
a threshold value for $z$. In the case of figure~\ref{f:strongdep}, we
have that $q$ dominates $r$ up to $k/n*100$.

We have computed strong dominance graphs for state of the art \FOSS{}
distributions, obtaining concise visual representations of clusters of
packages intertwined by strong dependencies;
Appendix~\ref{sec:sample-dominators} contains some of those graphs.

\section{Strong dependencies in Debian}
\label{sec:measures}

Due to the different properties of direct and strong dependencies,the
two measures of package sensitivity can differ substantially. To
verify that, as well as other properties of the underlying dependency
graphs, we have chosen Debian GNU/Linux as a case study.\footnotemark
The choice is not casual: Debian is the largest \FOSS{} distribution
in terms of number of packages (about $22'000$ in the latest stable
release) and, to the best of our knowledge, the largest
component-based system freely available for study.

\footnotetext{The data presents in this section, as well as what was
  omitted due to space constraints, are available to download from
  \url{http://www.mancoosi.org/data/strongdeps/}. The tools used to
  compute the data are released under open source licenses and are
  available from the Subversion repository at
  \url{https://gforge.info.ucl.ac.be/svn/mancoosi}.}

\begin{figure}
  \centering
  \includegraphics[width=0.75\columnwidth]{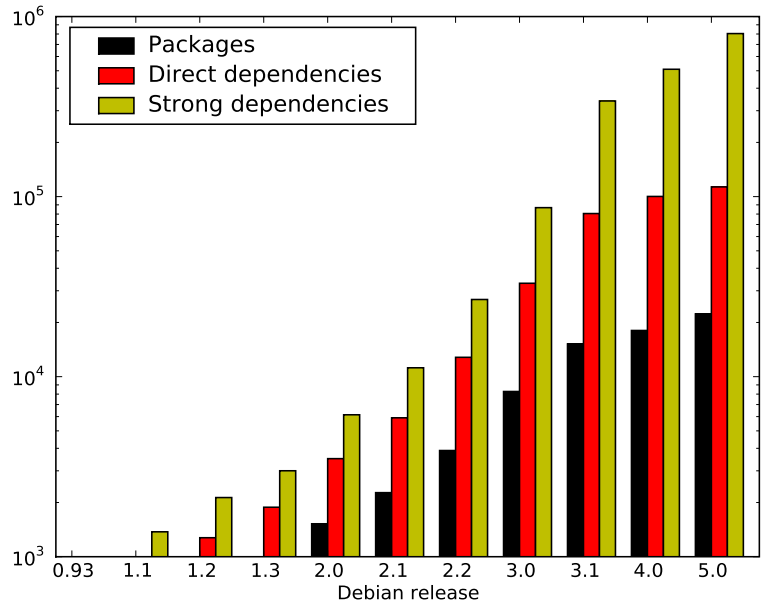}
  \caption{Evolution of packages, direct, and strong dependencies
    across Debian releases.\strut}
  \label{fig:debian-history}
\end{figure}

All stable releases of Debian have been considered, from 1994 to
February 2009. For each release the archive section \texttt{main} and
in particular the \texttt{i386} architecture has been considered; the
choices are justified by the fact that they identify both the most
used parts of Debian,\footnote{According to the Debian popularity
  contest, available at \url{http://popcon.debian.org}} and that they
are the only parts which have been part of all Debian releases and
hence can be better compared across years. The obtained archive parts
have been analysed by building both the direct and strong dependency
graphs; while the construction of the former is a trivial exercise,
the implemented efficient way of constructing the latter is discussed
in Section~\ref{sec:algo}. To build the direct dependency graph the
\texttt{Depends} and \texttt{Pre-Depends} inter-package relationships
have been
considered~\cite{debian-policy}.

Figure~\ref{fig:debian-history} shows the resulting evolution of the
number of graph nodes and edges across all Debian releases.  The size
of the distribution has grown steadily, yet super-linearly, across
most releases~\cite{robles-mining-compilations,baraona-spl}, but the
growth rate has decreased in the past two releases. As expected,
strong and direct sensitivity are not entirely unrelated, given that
the former is the semantic view of the latter, hence they tend to grow
together.

More precisely the total number of strong dependencies is higher, in
all releases, than the total number of direct dependencies. A partial
explanation comes from the fact that the strong dependency graph is a
transitive closed graph---property inherited by the underlying strong
dependency relationship---whereas the direct dependency graph is
not. Performing the transitive closure of the direct dependency graph
however would be meaningless, because the propagation rules of
disjunctive and conjunctive dependencies are not expressible simply in
terms of transitive arcs.

\begin{table}[b!]
  \centering
  \caption{Direct and strong sensitivity across Debian releases:
    correlation, mean, standard deviation.}
  \label{tab:sensitivity-stats}
  \begin{tabular}{c|c|c|c|c}
    \textbf{Rel.} & $\rho$
     & \multicolumn{1}{c|}{\DSENS{\cdot}}
     & \multicolumn{1}{c|}{\SENS{\cdot}}
     & \multicolumn{1}{c}{$\Delta$} \\
    \hline                                                                      
    \texttt{0.93} & 0.92 & 1.00,  $\sigma$2.79 &  1.05,   $\sigma$4.73 & 1.00,  $\sigma$4.00 \\
    \texttt{1.1}  & 0.93 & 1.70, $\sigma$13.91 &  2.90,  $\sigma$25.96 &  1.88, $\sigma$18.56 \\
    \texttt{1.2}  & 0.91 & 1.79, $\sigma$18.43 &  2.99,  $\sigma$32.27 &  1.73,  $\sigma$22.42 \\
    \texttt{1.3}  & 0.91 & 1.92, $\sigma$21.95 &  3.06,  $\sigma$38.24 &  1.69,  $\sigma$25.85 \\
    \texttt{2.0}  & 0.93 & 2.29, $\sigma$26.73 &  4.03,  $\sigma$50.88 &  2.50,  $\sigma$36.56 \\
    \texttt{2.1}  & 0.94 & 2.60, $\sigma$34.92 &  4.93,  $\sigma$64.55 &  2.93,  $\sigma$46.61 \\
    \texttt{2.2}  & 0.92 & 3.29, $\sigma$44.24 &  6.89,  $\sigma$90.47 &  4.88,  $\sigma$68.76 \\
    \texttt{3.0}  & 0.92 & 3.99, $\sigma$59.25 & 10.49, $\sigma$131.5  &  8.02,  $\sigma$92.34 \\
    \texttt{3.1}  & 0.92 & 5.29, $\sigma$91.42 & 22.36, $\sigma$282.0  & 19.39 , $\sigma$246.4 \\
    \texttt{4.0}  & 0.92 & 5.55, $\sigma$85.10 & 28.23, $\sigma$352.4  & 24.50 , $\sigma$313.9 \\
    \texttt{5.0}  & 0.93 & 5.07, $\sigma$86.16 & 36.05, $\sigma$480.3  & 32.55 , $\sigma$440.1 \\
  \end{tabular}
\end{table}

We have studied the apparent correlation between strong and direct
dependencies analysing the respective sensitivity measures for each
release. Table~\ref{tab:sensitivity-stats} confirms the correlation
and gives some statistical data about package sensitivity. The first
column is the Spearman $\rho$ correlation index,\footnote{The
  statistical info for the first two rows are possibly not relevant,
  due to the small size of the two releases.} a commonly used
non-parametric correlation index that is not sensible to exceptional
values~\cite{fenton_software-metrics}. An
index between 0.5 and 1.0---in all the releases we have $\rho\in
[0.91,0.94]$---is commonly interpreted as a strong correlation between
the two variables. The more common correlation index $r$ for the same
set of data (not shown in the table) gives consistently a value of $0.55$: 
the huge difference among $\rho$ and $r$ indicates that the few exceptional
values in the data series have really high weight, and we will see, when
looking at some of these exceptional value that this is indeed the case.

The remaining columns show mean and standard
deviation for, respectively, direct sensitivity, strong sensitivity,
and $\Delta = \SENS{p} - \DSENS{p}$. In particular we note an
increasingly high standard deviation in latest Debian releases, which
hints that there is an increasing number of peaks.


\begin{figure}[t!]
  \center
  \includegraphics[width=0.80\columnwidth]{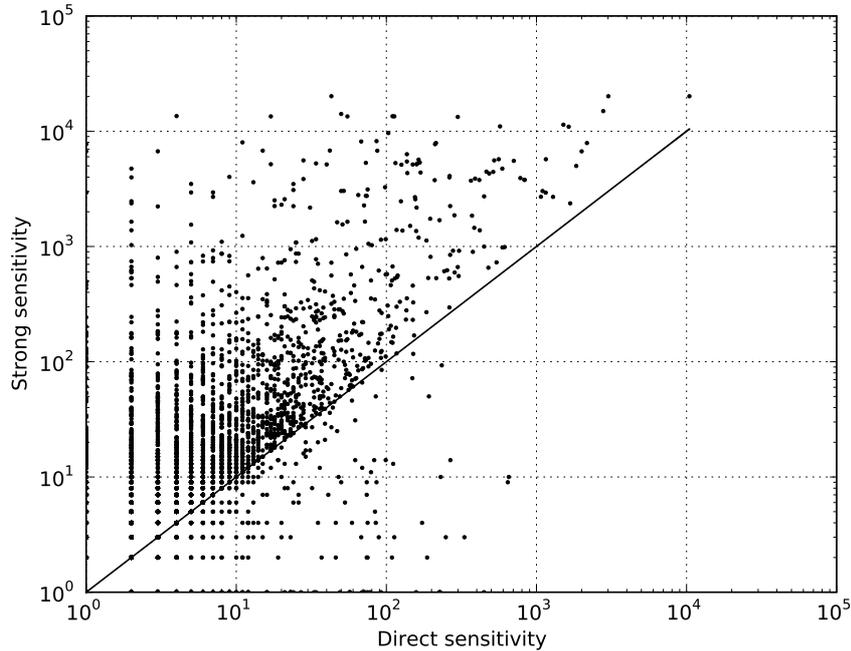}
  \caption{Correlation between strong and direct sensitivity in Debian
    5.0 ``Lenny''.\strut}
  \label{fig:lenny-correlation}
\end{figure}

Figure~\ref{fig:lenny-correlation} shows in more detail the
correlation phenomenon for Debian 5.0 ``Lenny'', the latest and
largest Debian release. The figure plots strong vs direct dependencies
for each package in the release. In most cases, strong sensitivity is
higher than direct sensitivity, yet close: 82.9\% of the packages fall
in a standard deviation interval from the mean of $\Delta$; the next
percentile ranks are 97.4\% for two standard deviations, and 99.8\%
for three. The remaining cases allow for important exceptions of
packages with very high strong sensitivity and very low direct
sensitivity. Such exceptions are extremely relevant: metrics built on
direct sensitivity only would totally overlook packages with a huge
potential impact.

\begin{table}[t!]
  \centering
  \caption{Packages from Debian 5.0, sorted by the difference between
    strong and direct impact sets.}
  \label{tab:lenny-underestimate}
  \begin{tabular}{r|l|r|r|r}
    \multicolumn{1}{c|}{\#}
     & \multicolumn{1}{c|}{\textbf{Package}}
     & \multicolumn{1}{c|}{\DSENS{p}}
     & \multicolumn{1}{c|}{\SENS{p}}
     & \multicolumn{1}{c}{$\SENS{p}-\DSENS{p}$} \\
    \hline
    \INCR & \texttt{gcc-4.3-base} & 43 & 20128 & 20085 \\
    \INCR & \texttt{\texttt{libgcc1}} & 3011 & 20126 & 17115 \\
    \INCR & \texttt{libselinux1} & 50 & 14121 & 14071 \\
    \INCR & \texttt{lzma} & 4 & 13534 & 13530 \\
    \INCR & \texttt{coreutils} & 17 & 13454 & 13437 \\
    \INCR & \texttt{dpkg} & 55 & 13450 & 13395 \\
    \INCR & \texttt{libattr1} & 110 & 13489 & 13379 \\
    \INCR & \texttt{libacl1} & 113 & 13467 & 13354 \\
    \INCR & \texttt{perl-base} & 299 & 13310 & 13011 \\
    \INCR & \texttt{libstdc++6} & 2786 & 14964 & 12178 \\
    \INCR & \texttt{libncurses5} & 572 & 11017 & 10445 \\
    \INCR & \texttt{debconf} & 1512 & 11387 & 9875 \\
    \INCR & \texttt{libc6} & 10442 & 20126 & 9684 \\
    \INCR & \texttt{libdb4.6} & 103 & 9640 & 9537 \\
    \INCR & \texttt{zlib1g} & 1640 & 10945 & 9305 \\
    \INCR & \texttt{debianutils} & 86 & 8204 & 8118 \\
    \INCR & \texttt{libgdbm3} & 68 & 8148 & 8080 \\
    \INCR & \texttt{sed} & 11 & 8008 & 7997 \\
    \INCR & \texttt{ncurses-bin} & 1 & 7721 & 7720 \\
    \INCR & \texttt{perl-modules} & 214 & 7898 & 7684 \\
    \INCR & \texttt{lsb-base} & 211 & 7720 & 7509 \\
    \INCR & \texttt{libxdmcp6} & 15 & 6782 & 6767 \\
    \INCR & \texttt{libxau6} & 42 & 6795 & 6753 \\
    \INCR & \texttt{libx11-data} & 1 & 6693 & 6692 \\
    \INCR & \texttt{libxcb-xlib0} & 3 & 6695 & 6692 \\
    \INCR & \texttt{libxcb1} & 87 & 6778 & 6691 \\
    \INCR & \texttt{x11-common} & 137 & 6317 & 6180 \\
    \INCR & \texttt{perl} & 2169 & 7898 & 5729 \\
    \INCR & \texttt{libmagic1} & 28 & 5585 & 5557 \\
    \INCR & \texttt{libpcre3} & 164 & 5668 & 5504 \\
    \multicolumn{5}{c}{\ldots} \\
  \end{tabular}
\end{table}

\subsection{Strong vs direct sensitivity: exceptions}

It's time now to look at some of these exceptional cases to see how
relevant they are. Table~\ref{tab:lenny-underestimate} lists the top
30 packages of Lenny having the largest $\Delta$.

\PKG{libc6} is the package shipping the C standard library which is
required, directly or not, by almost all applications written or
otherwise linked to the C programming language. About a half of all
the packages in the distribution depends \emph{directly} on
\PKG{libc6}, as can be seen in row 13 of the table, but almost all
packages in the archive cannot be installed without it, as the strong
sensitivity of \PKG{libc6} is 20'126, on a total of 22'311
packages. In this case direct sensitivity does not inhibit identifying
the package as a sensitive one, though, even if it underestimates
widely its importance.

Now consider row 1 of Table~\ref{tab:lenny-underestimate}:
\PKG{gcc-4.3-base}, which is a package without which \PKG{libc6}
cannot be installed. It is the package with the largest $\Delta$,
having direct sensitivity of only 43 and strong sensitivity of
20'128. Ranking its sensitivity with the direct metric would have led
to completely miss its importance: a bug into it can potentially
affect all packages in the distribution. Note however that
\PKG{gcc-4.3-base} is not a direct dependency of \PKG{libc6}, showing
once more that to grasp this kind of inter-package relationships the
semantics, rather than the syntax, of dependencies must be put into
play.

In the second row, \PKG{libgcc1} shows a similar pattern, being this
time a direct dependency of \PKG{libc6}. The third row and many others
in the table show more complex patterns. Ordering packages only
according to sensitivity might lead to oversee other important
characteristic. Possibly the most extreme cases are those of
\PKG{ncurses-bin} and \PKG{libx11-data}, which are mentioned just once
in all the explicit dependencies, and yet are really necessary for
several thousand other packages.

We believe this is sufficiently conclusive evidence to totally dismiss,
from now on, any analysis based on the syntactic, direct dependency graph,
when considering component based systems with complex dependency languages.
  
\subsection{Using strong dominance to cluster data}

Now we turn to the problem of presenting the sensitiveness information in
a relevant way to a Quality Assurance team: we could simply print a list
of package names, ordered by their sensitiveness; this would give a result
quite similar to that of table~\ref{tab:lenny-underestimate} above, just 
dropping the first and fourth column. A smart Debian maintainer will surely
spot the fact that \PKG{gcc-4.2-base}, \PKG{libgcc1} and \PKG{libc6} are
related and would look at them together, but it would be difficult to 
see relationships among the other packages in the list, even if we can
see that many packages have impact sets of similar size.\\

\begin{figure*}[t!]
  \includegraphics[width=0.98\textwidth]{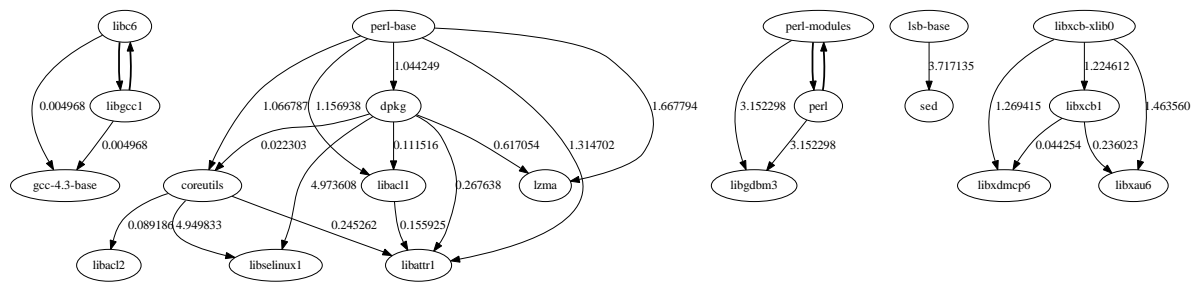}
  \caption{Dominance relations among the topmost 20 sensitive packages}
  \label{fig:clustering}
\end{figure*}

Here is where our definition of relative strong dominance comes into
play, allowing to build meaningful clusters that provide sensible
information to the maintainers: Figure~\ref{fig:clustering} shows the
graph of relative strong domination between the first $20$ packages of
Table~\ref{tab:lenny-underestimate}. Bold edges show strong domination
as defined in Definition~\ref{def:dominance}.  Normal edges show
relative domination, where the install sets of the two packages almost
fully overlap, apart from a few packages (edges are labelled with the
percentage $z$ of Definition~\ref{def:reldominance}).

This figure shows clearly that it is possible to isolate five clusters of
related packages with similar sensitivity values; some of them may look
surprising at first sight to a Debian maintainer, and evident after a
little time spent exploring the package metadata: this actually confirms
the real value of this way of presenting data.

\subsection{Debian is a small world}

We expected the strong dependency graph to retain the small world
characteristics previously established for the direct dependency
graph~\cite{labelle_os_networks}, but this required some extra effort
to get sensible results: indeed, computing clustering coefficients and
other similar measures on the strong dependency graph will yield very
different values (as the strong dependency graph is transitive), so we
first built a detransitivised version of the strong dependency graph,
and computed the usual small world measures on it.

Note that, since the strong dependency graph contains some cycles, the
graph obtained by detransitivising it is not unique. The differences
are however minor enough to not alter the overall results.

\begin{table}[b!]
  \label{table:swc}
  \caption{Small-world characteristics for Debian 5.0.}
  \centering
  \begin{tabular}{l|r|r}
     & \multicolumn{1}{p{5cm}|}{\textbf{Direct dependency graph}}
     & \multicolumn{1}{p{5cm}}{\textbf{Strong depedency graph}} \\
    \hline
    \emph{Vertices} & 22 311 & 22 311 \\
    \emph{Edges} & 107 796 & 40 074 \\
    \emph{Average degree} & 4.83 & 1.80 \\
    \hline
    \emph{Clustering coeff.} & 0.41 & 0.39 \\
    \emph{Average distance} & 3.18 & 2.86 \\
    \hline
    \emph{Components (WCCs)} & 1425 & 2809 \\
    \emph{Largest WCC} & 20 831 & 19 200 \\
    \hline
    \emph{Density} & 0.00022 & 0.000081
  \end{tabular}
\end{table}

The clustering coefficient and average path length of the
detransitivised graph are, though slightly smaller, well within the
range of small-world networks. More than half the edges of the
syntactic graph have disappeared, but this has not significantly
affected either the graph clustering or the path length. The relevant
statistics are summarised in Table~\ref{table:swc}.

Some further notes from Table~\ref{table:swc}. First, both graphs
contain one enormous (weakly connected) component, next to which all
other components are of insignificant size (for the direct graph,
there are 1'480 remaining packages in 1'424 components, which would
make their average size just above 1; the ratio is similar for the
strong graph). Second, when we look at the density of both graphs (the
number of edges in the graph divided by the maximum possible number of
edges), we see that both graphs are extremely sparse.

\section{Efficient computation}
\label{sec:algo}

It is not evident that strong dependencies as defined in
Section~\ref{sec:strongdeps} are actually tractable in practise: from
previous results~\cite{EdosAse06,EDOS-FRCSS06} it is known that
checking installability of a package (or co-installability of a set of
packages) is an NP-complete problem. Even if in practise checking
installability turns out to be tractable on real-world problem
instances, the sheer number of instances that computing strong
dependencies may require in general makes the problem in principle
much harder. We start by observing that the problem of determining
strong dependencies is decidable.

\PRETHM
\begin{proposition}[Decidability]\label{prop:sddec}
  Strong dependencies for packages in a \emph{finite} repository $R$
  are \emph{computable}.
\end{proposition}
\begin{proof}
  Since $R$ is finite, the set of all installations is also finite.
  Among these installations, finding the healthy one is just a matter
  of verifying locally the dependency relations. Then, for each $p$
  and $q$, it is enough to check all healthy installations to see
  whether $q$ is present whenever $p$ is.
\end{proof}
\POSTTHM

If we want to know if a particular packages $p$ strongly depends on
$q$ in a repository $R$ however, the argument used in the proof of
decidability leads to an algorithm that has exponential worst-case
complexity in the size $n$ of a repository $R$. One possible algorithm
to find \emph{all} strong dependencies in a repository $R$ is as
follows.

\begin{algorithmic}
\REQUIRE $R \neq \emptyset$
\STATE $\mathtt{strongdeps} \leftarrow \emptyset$
\FORALL {$p,q\in R$}
  \IF{$\mathtt{strong\_dependency}(p,q,R)$}
    \STATE $\mathtt{strongdeps} \leftarrow \mathtt{strongdeps} \cup \{p,q\}$
  \ENDIF
\ENDFOR
\textbf{return} $\mathtt{strongdeps}$
\end{algorithmic}

Where the function $\mathtt{strong\_dependency}$ uses a SAT solver to check
whether it is possible to install $p$ without installing $q$ (in repository
$R$). This algorithm requires checking $n^2$ SAT instances, which is unfeasible
with $n\approxeq 22'000$. We need to look for an optimised approach; the
following remark is the key observation.

\PRETHM
\begin{remark}[Reducing the search space]
  All packages $q$ on which a given package $p$ strongly depends are
  included in \emph{any} installation of $p$. Furthermore, if a
  package $p$ conjunctively depends on a package $q$, then $q$ is a
  strong dependency of $p$.
\end{remark}
\POSTTHM

This leads to the following improved algorithm that strongly relies on
the notion of installation sets and the property of transitivity of
strong dependencies. 

\begin{algorithmic}
\FORALL {$p \in R$}
  \STATE $\mathtt{strongdeps} \leftarrow \mathtt{strongdeps} \cup
  \mathtt{conj\_dependencies}(p,R)$
\ENDFOR
\FORALL {$p \in R$}
  \STATE $S \leftarrow \mathtt{install}(p,R)$
  \FORALL {$q \in S$}
    \IF {$(p,q) \not\in \mathtt{strongdeps} \land
      \mathtt{strong\_dependency}(p,q,R)$}
    \STATE $\mathtt{strongdeps} \leftarrow \mathtt{strongdeps} \cup \{p,q\}$
    \ENDIF
  \ENDFOR
\ENDFOR
\textbf{return} $\mathtt{strongdeps}$
\end{algorithmic}

The function $\mathtt{conj\_dependencies}(q,R)$ returns all packages
in $R$ that are connected to $q$, considering only conjunctive
paths. We add to the $\mathtt{strongdeps}$ set all couples $(p,q)$
such that there exists a conjunctive path between $p$ and $q$, and
then for all remaining packages in the install set of $p$, we check if
there is a strong dependency using the SAT solver.

On one hand, the analysis of the structure of the repositories shows
that it is in practice possible to find installation sets that are
quite small. Considering only the installation set for a given package
drastically reduces the number of calls to the SAT solver. On the
other hand, since the large majority of strong dependencies can be
derived directly from conjunctive dependencies, building the graph of
conjunctive dependencies beforehand can further reduce the computation
time.

In our experiments, calculating the strong dependency graph and
sensitivity index for about $22'000$ packages takes about 5 minutes on
a modern commodity Unix workstation.\footnote{Intel Xeon 3 GHz
  processor, 3 Gb of memory}

\section{Applications}
\label{sec:applications}

The given notions of strong dependency, impact set, sensitivity, and
strong dominance can be used to address issues showing up in the maintenance
of large component repositories. In particular, we have identified two
areas of application: repository-wide Quality Assurance (QA) and
upgrade risk evaluation for user machines.

\paragraph{Quality Assurance}
\FOSS{} distribution the size of Debian are not easily inspectable by
hand, without specific tools. The work of release managers in such
scenario is about maintaining a coherent package repository, i.e., in
which each package is installable in at least one healthy
installation. Such repositories are usually not built from scratch,
but rather evolve from an unstable state to a stable one which is
periodically released as the new major release of the
distribution. Day to day maintenance of the repository includes
actions such as adding packages to the repository (e.g., newly
packaged software, or new releases) as well as removing them (e.g.,
superseded softwares or sub-standard quality packages which are not
considered suitable for releasing). Quality assurance is meant to spot
repository-wide incompatibilities or sub-standard quality packages,
according to various criteria.

In such ecosystems, removing a package can have non-local effects
which are not evident by just looking at the direct dependencies of
the involved packages. For instance, removing a package $p$ such that
several packages depends on $p~|~q$ might be appropriate only if $q$
is installable in the archive. The strong dependency graph can be used
to detect similar cases efficiently. Once the graph has been
computed---and Section~\ref{sec:algo} showed that the cost is
affordable even for large distributions---detecting if a package is
removable in isolation reduces to check whether its node has inbound
edges or not. If really needed, following inbound edges can help
building sets of packages removable as a whole.

In the same context, sensitivity can be used to decide when to freeze
packages during the release process (decision currently delegated to
folklore): the higher the sensitivity, the sooner a package should be
frozen. Sensitivity can also be used to activate heuristic warnings in
archive management tools when apparently innocuous packages are acted
upon: attempting to remove or otherwise alter \PKG{gcc-4.3-base} at
the end of the Lenny release process (see
Table~\ref{tab:lenny-underestimate}) would have surely been an error,
in spite of the few packages mentioning it directly in their
dependencies.


\paragraph{Upgrade risk evaluation}
System administrators of machines running \FOSS{} distributions would
like to be able to judge the risks of a certain upgrade. Risk
evaluation not necessarily in the sense of deciding whether or not to
perform an upgrade---not performing one is often not an option, due to
the frequent case of upgrades that fix security vulnerability. Upgrade
risk evaluation is important to allocate suitable time slots to deploy
upgrade plans proposed by package managers: the riskier the upgrade,
the longer the time slot that should be planned for it.

The general principle we propose is that a package that is not
strongly depended upon by other packages is relatively safe to
upgrade; conversely, a package that is needed by many packages on the
system might need some safety measures in case of problems (backup
servers, \ldots). However this measure should be computed in relation
to the actual user installation and not as an absolute value with
respect to the distribution such as plain impact sets. Once the strong
dependency graph of a user installation has been computed, the legacy
package manager can be used to find upgrade plans as usual. On that
plan the overall upgrade sensitivity can then be computed by summing
up the size of the \emph{installation impact sets} of all packages
touched by the proposed plan; where the installation impact set of a
package $p$ is defined as the intersection of the strong impact set
with the local installation.

The strong dependency graph used for risk evaluation must be the one
corresponding to the distribution snapshot which was known
\emph{before} planning the upgrade. This is because we want to
evaluate the risks with respect to the current installation, not to a
future potential one in which package sensitivity can have
changed. The maintenance of such graph on user machines is
straightforward and can be postponed to after upgrade runs have been
completed, in order to be ready for future upgrades.

Note that in this way, what is computed is an under approximation of
the upgrade risk measure. For example consider the following scenario:
a package $p$ having \lstinline[style=debctrl]$Depends: q | r$, and a
healthy installation $I = \{p,q\}$. The direct dependencies of $p$
entail no strong dependency, but in the given installation $q$ has
been ``chosen'' to solve $p$ dependencies. Even if $p \not\in
\IS{q}{R}\cap I$, an upgrade of $q$ in that specific installation has
potentially an impact on $p$. The under approximation is nevertheless
sound---i.e., all packages in the installation impact set are
installed.

\paragraph{Release upgrades}
A particular case of upgrade are the so called \emph{release upgrades}
(or distribution upgrades) which are performed periodically to switch
from an older stable release of a given distribution to a newer
one. The relevance of such upgrades is that they usually affect almost
all of the packages present in user installation. Such kind of
upgrades are usually already performed wisely by system administrators
devoting large enough time slots.

During release upgrades system administrators are often faced with the
choice of whether to switch to a new major version of some available
software or to stay with an older, legacy one. For instance, one can
have the choice to switch to the Apache Web server 2.x series, or to
stay with Apache 1.x. The upgrade is not forced by strict package
versioning by either offering packages with different names
(e.g. \texttt{apache1} vs \texttt{apache2} in Debian and its
derivatives) or by avoiding explicit conflicts among the two set of
versions (as it happens in RPM-based distributions). The choice is
currently not technically well assisted: if \texttt{apache2} is
tentatively chosen, the package manager will propose to upgrade all
involved packages to the most recent version without highlighting
which upgrades are \emph{mandatory} to fulfil dependencies and which
are not.

While this is a deficiency of state of the art solving
algorithms~\cite{mancoosi-debconf8}, strong dependencies offer a cheap
technical device to work around the problem with current solvers. It
is enough to compute the strong dependency graph of both distributions
and, in particular, the strong dependencies of the two (or more)
involved packages. Then, by taking the difference of the strong
dependencies in the new and in the old graph, the list of package
which must be forcibly upgraded to do the switch is obtained. All such
\emph{forced upgrades} can then be presented to the administrator to 
better guide her or his choice.

\paragraph{Architecture visualization and monitoring}
As the KDE cluster in Figure \ref{fig:sarge-kde} shows, there are
cases where the distributed development of a universe of components
can lead to partially overlapping, duplicated components, which mangle
the global architecture of the distribution; strong dominator graphs
extract significant information from the huge dependency graph, and
help identify in a very visual way the areas that need to be analysed
in detail; it is natural to try and improve this line of research by
coupling the information found via strong dominators to other information
available on the software components, like the developer team members, or the
upstream source code.

\section{Related works}
\label{sec:related}

We have identified several interesting articles dealing with issues
related to the topics we address. In the area of complex networks,
\cite{labelle_os_networks,maillart_empirical_2008} used \FOSS{}
distributions as case studies. The former is the closest to our
focus, as it studies the network structure obtained from Debian
inter-package relationships, showing that it is small-world, as the
node connectivity follows a near power-law distribution. However,
the analysis is performed on the direct dependency graph, which we
have shown misses the real nature of dependencies.

We could not get more information on how the data
of~\cite{labelle_os_networks} has been computed, as the snapshot of
Debian used there comes from late 2004, and is no longer available
in the Debian archives; based on the figures presented in the paper,
and our analysis of the closest Debian stable distribution, we
conclude that their analysis dropped all information about
\texttt{Conflicts} and \texttt{Pre-Depends}. As a consequence, the
figures produced for what is called in the paper ``the 20 most
highly depended upon packages'' falls extremely short of reality:
\texttt{libc6} is crucial for 3 times more packages than what is
reported, and other critical packages such as \texttt{gcc-4.3-base}
are entirely missed.

In the area of quality assurance for large software projects, many
authors correlate component dependencies and past failure rates in
order to predict future
failures~\cite{zimmermann-predicting-defects,nagappan_churn,zeller-predicting}.
The underlying hypothesis is that software ``fault-proneness'' of a
component is correlated to changes in components that are tightly
related to it.  In particular if a component $A$ has many
dependencies on a component $B$ and the latter changes a lot between
versions, one might expect that errors propagates through the
network reducing the reliability of $A$.  A related interesting
statistical model to predict failures over time is the ``weighted
time damp model'' that correlates most recent changes to software
fault-proneness~\cite{graves-predicting-history}. Social network
methods~\cite{hanneman05-introduction} were also used to validate
and predict the list of \emph{sensitive} components in the Windows
platform~\cite{zimmermann-predicting-defects}.

Our work differs for two main reasons. First, the source of
dependency information is quite different. While dependency
analysing for software components is inferred from the source code,
the dependency information in software distributions are formally
declared and can be assumed to be, on the average, trustworthy as
reviewed by the package maintainer. Second, \FOSS{} distributions
still lack the needed data to correlate upgrade disasters with
dependencies and hence to create statistical models that allow to
predict future upgrade disasters. In more detail, the \FOSS{}
ecosystem is really fond of public bug tracker systems, but
generally lacks explicit logging of upgrade attempts and a way to
associate specific bugs to them. One of the goal of the
\MANCOOSI\footnote{\url{http://www.mancoosi.org}} project---in which
the authors are involved---is to create a corpus of upgrade problems
which will be a first step in this direction.

The key idea behind the notion of sensitivity can be seen as a
direct application of the evaluation of ``disease spreading speed''
in small world networks~\cite{watts_collective_1998}: the higher the
sensitivity, the larger the impact sets, the higher the (potential)
bug spreading speed. The semantic definition of impact sets is
crucial in this analysis: using the direct dependency graph would
give no guarantee about which components will be effectively
installed and therefor help bug spreading.

\section{Conclusion and future work}
\label{sec:outro}

This paper has introduced the novel notions of \emph{strong
  dependencies} between software components, and of \emph{sensitivity}
as a measure of how many other components rely on its availability; and
we have introduced \emph{strong dominance} as a means of ordering and grouping
components with similar sensitivity into meaningful clusters. We
have shown concretely on a large scale real world example 
that such notions are better suited to describe true
inter-component relationships than previous studies, which were solely
based on the analysis of the syntactic (or direct) dependency
graph. The main applications of these new notions are tools for
quality assurance in large component ecosystems and upgrade risk
evaluation.

The new notions have been tested on one of the largest known
component-based system: Debian GNU/Linux, a popular \FOSS{}
distribution. Historical analysis of Debian strong and direct
dependency graphs have been performed. Empirical evidence shows that,
while the two notions are generally correlated, there are several
components on which they give huge differences, with direct
dependencies entirely missing key components that are rightly
identified by strong dependencies.  We believe the case shown in this
paper is strong enough to totally dismiss, in the future, measures
built on direct dependencies as soon as the dependency language is as
expressive as the one used for \FOSS{} distributions.

We hence strongly advocate the evaluating of sensitivity on top of strong
dependencies, and we have shown clearly how clustering components according to
the notion of strong dominance allows to build a meaningful presentation of
data, and uncover deep relationships among components in a repository.

Despite the theoretical complexity of the problem, and the sheer size of modern
component repositories, we have succeeded in designing a simple optimised
algorithm for computing strong dependencies that performs very well on real
world instances, making all the measures proposed in this paper not only
meaningful, but actually feasible.

Previous studies on network properties---such as small world
characteristics---have been redone on the Debian strong dependency graph,
showing that it stays small world.

Future works is planned in various directions. First of all the notion
of installation impact set needs to be refined. While it is clear that
the strong impact set is an under approximation of it, it is not clear
how to further refine it. On one hand we want to get closer to the
actual set of potentially affected packages on a given machine. On the
other it is not clear, for a package $p$ depending on $q~|~r$ to which
extent \emph{both} packages should be considered as potentially
impacted by a bug in $p$. It appears to be a limitation in the
expressiveness of the dependency language which does not state an
order between $q$ and $r$, but needs further
investigation. Interestingly enough, the implicit syntactic order
``$p$ before $q$'' is already taken into account by some distribution
tools such as build daemons and is hence worth modelling.

Distributions like Debian use a staged release strategy, in which two
repositories are maintained: an ``unstable'' and a ``testing''
one. Packages get uploaded to unstable and migrate to testing when
they satisfy some quality assurance criteria, including the goal of
maintaining testing devoid of uninstallable packages. Current modelling
of the problem is scarce and implementations rely on empirical
package-by-package migration attempts. We believe that the notion of strong dependency
and the clusters entailed by strong dominance can help in identifying clusters
of packages which should forcibly migrate together.

\section{Acknowledgements}
\label{sec:ack}

The authors would like to thank Yacine Boufkhad, Ralf Treinen and Jer\^ome Vouillon
for many interesting discussions on these issues.


\bibliography{biblio}
\bibliographystyle{abbrv}
\newpage

\appendix
\renewcommand{\thesection}{\Alph{section}} 
\section{Case Study: Evaluation of debian structure}
\label{sec:sample-dominators}

The strong dominance graph can be used to extract and visualize
cluster of components with similar sensitivity that are hidden when
looking directly at the much bigger direct (or strong) dependency
graph.  To showcase this application of strong dominance, we present
the corresponding graphs for various Debian releases (see
Figures~\ref{fig:debian-strong-dom-series1},
\ref{fig:debian-strong-dom-series2},
\ref{fig:debian-strong-dom-series3},
\ref{fig:debian-strong-dom-series4}, and
\ref{fig:debian-strong-dom-series5}). In Figure \ref{fig:sarge-kde} we
highlight the structure of the KDE cluster in Debian 3.1.

Early releases (Figure~\ref{fig:debian-strong-dom-series1}) have small
dominance graphs hinting that the Debian distribution was composed,
back then, by a small number of loosely coupled package clusters. From
Debian 2.2 a growth of the strong dominance graph can be observed.

The following graphs are generate removing all non-trivial (i.e.,
strictly larger than 2 components) strong dominance clusters and with
a fuzzy index of $5\%$. All images are embedded in the electronic
format of this document in vectorial format allowing to ``zoom in''
the otherwise huge graph. The source files (in sgv and dot format) can
be retrieved from the mancoosi website at
\url{http://www.mancoosi.org}.

\begin{figure}[b!]
  \centering
  \subfloat[][]{\includegraphics[width=0.40\textwidth]{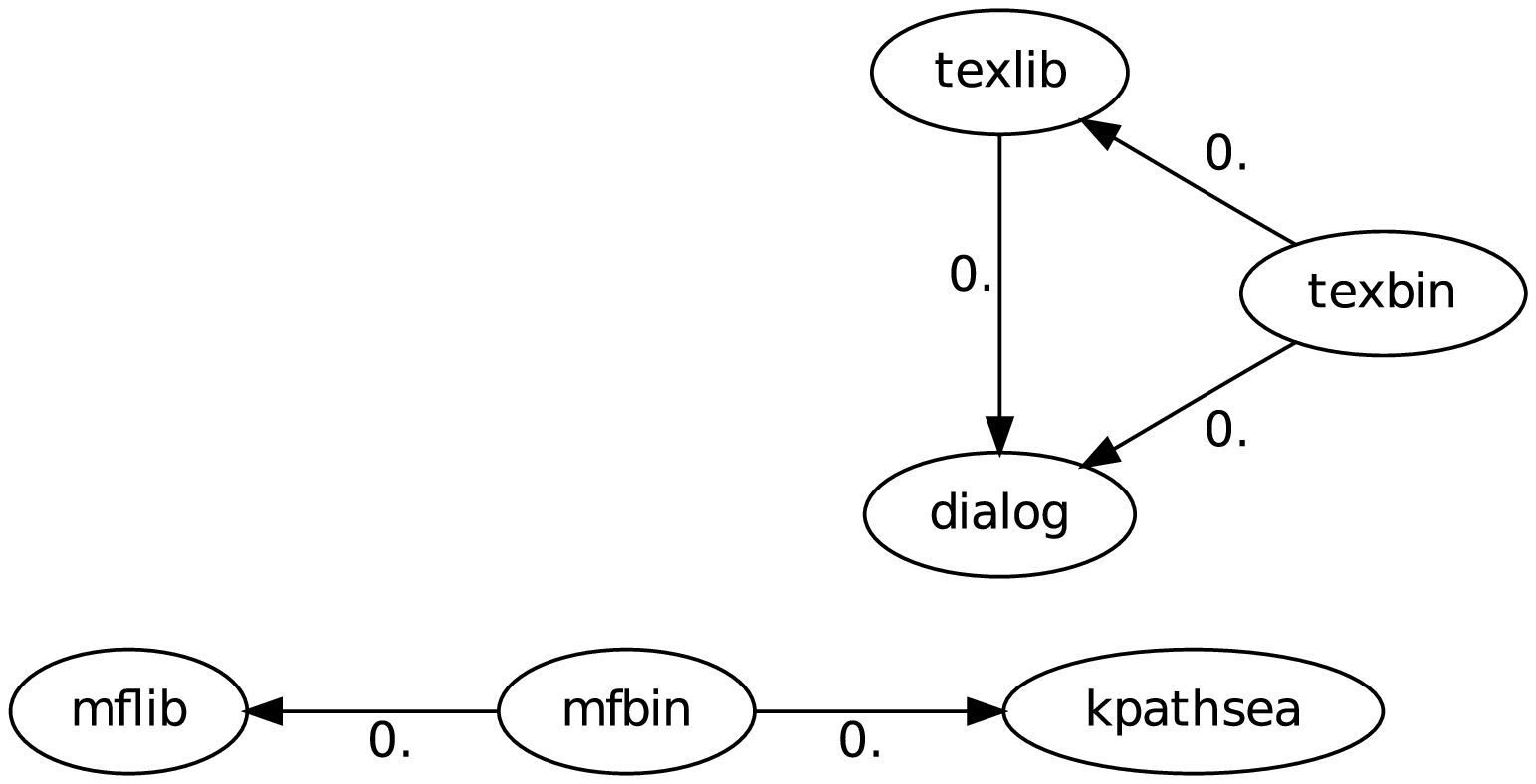}}%
  \qquad
  \subfloat[][]{\includegraphics[width=0.40\textwidth]{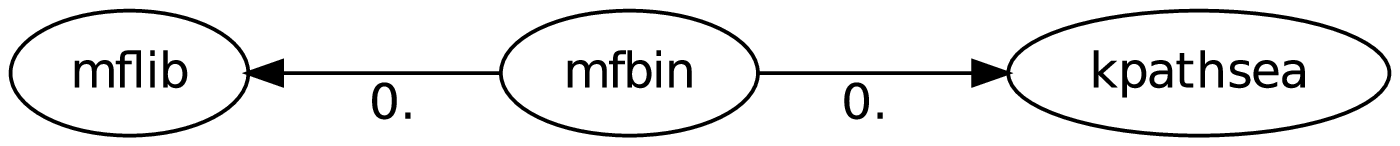}}%
  
  \subfloat[][]{\includegraphics[width=0.40\textwidth]{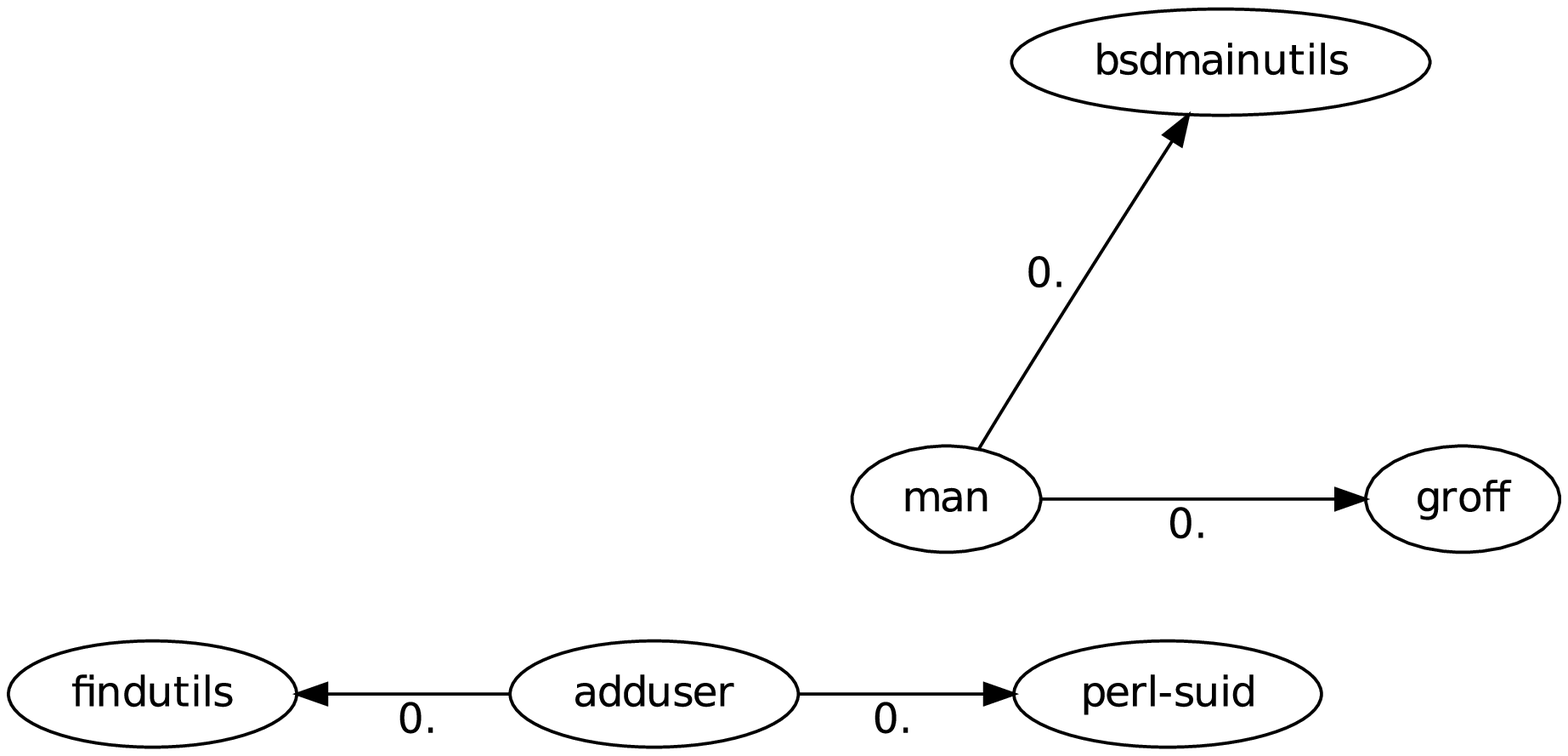}}%
  \qquad
  \subfloat[][]{\includegraphics[width=0.40\textwidth]{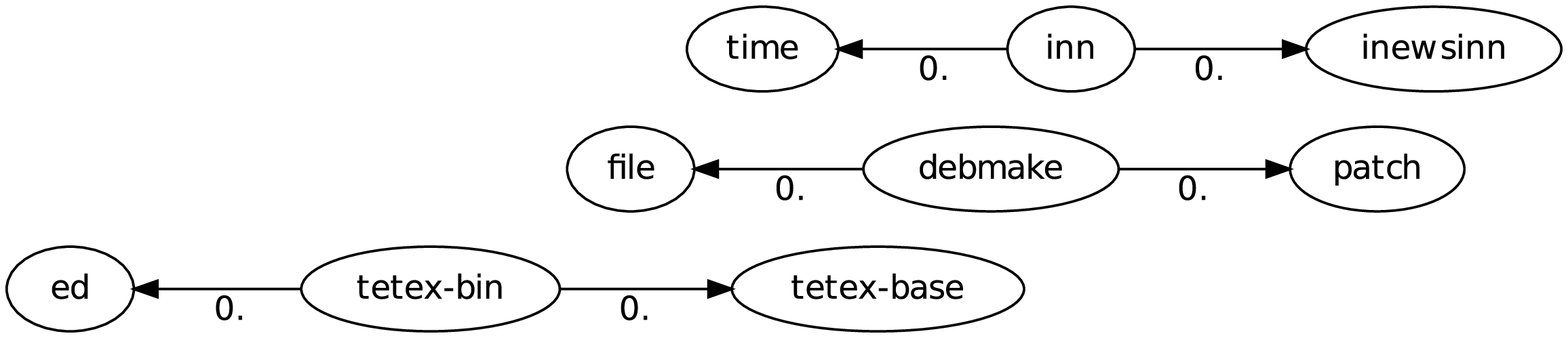}}%
  
  \subfloat[][]{\includegraphics[width=0.40\textwidth]{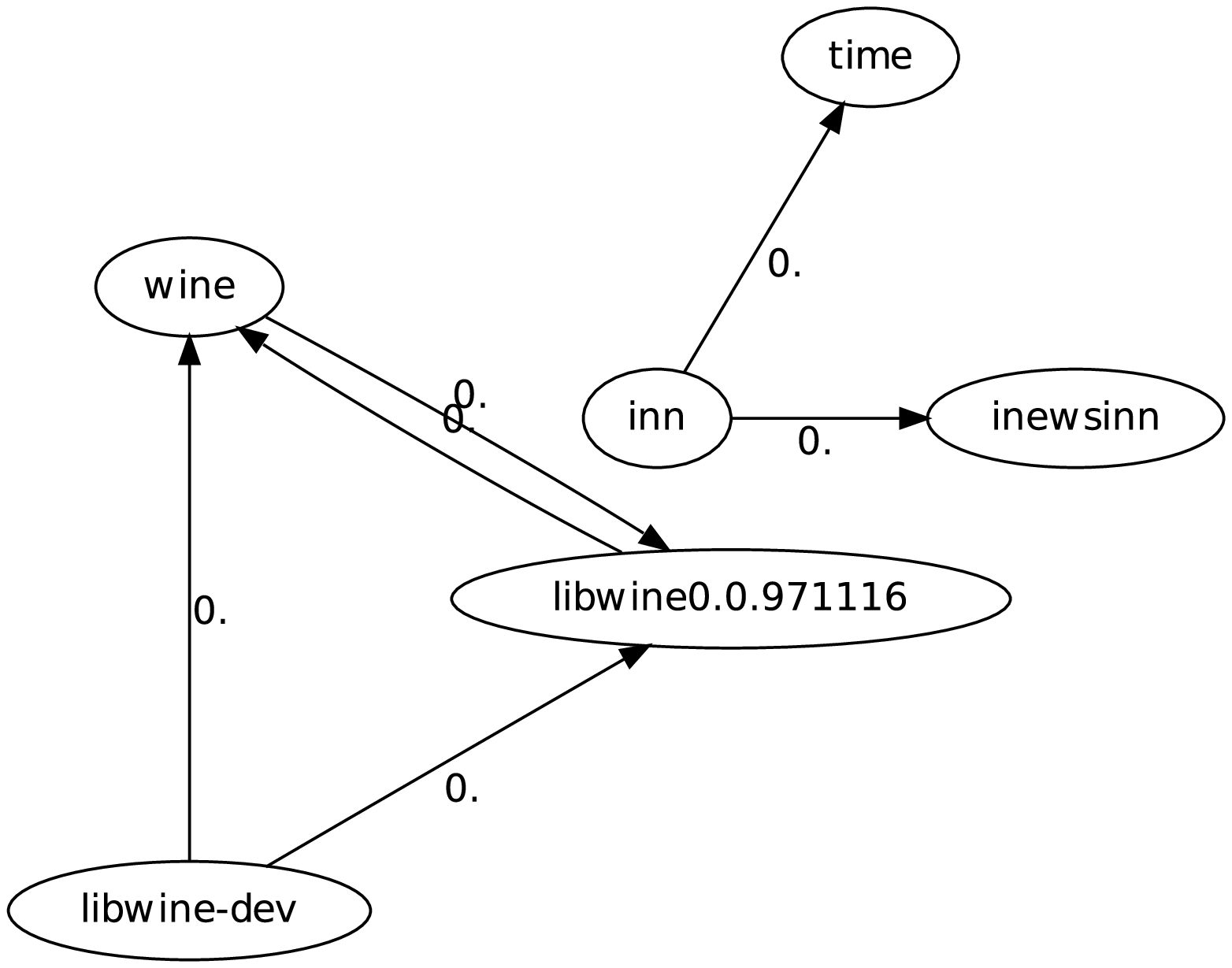}}%
  \qquad
  \subfloat[][]{\includegraphics[width=0.40\textwidth]{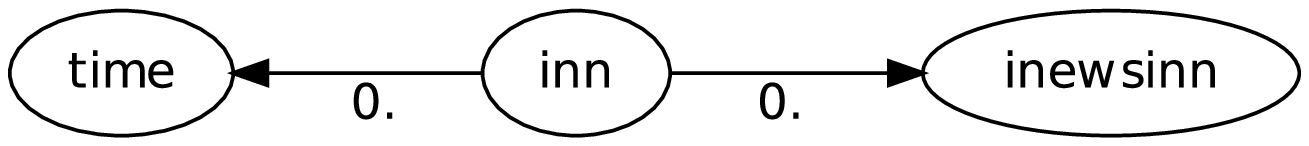}}%
  
  \caption{Strong dominance graphs for Debian 0.93R6, 1.1, 1.2, 1.3.1, 2.0, 2.1}%
  \label{fig:debian-strong-dom-series1}%
\end{figure}

\begin{figure}%
  \centering
  \subfloat[][]{\includegraphics[width=0.70\textwidth]{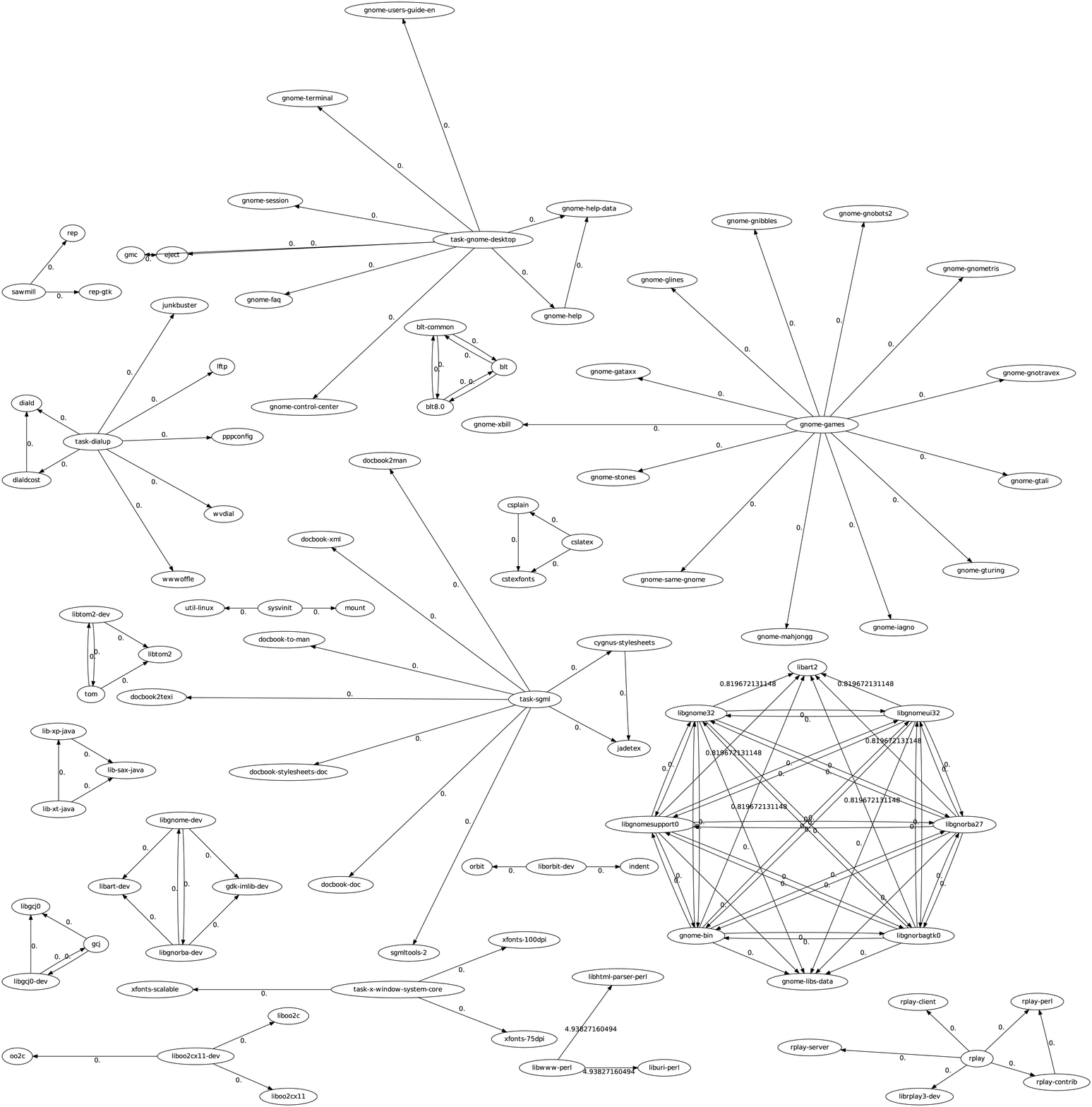}}%
  \qquad
  \subfloat[][]{\includegraphics[width=0.70\textwidth]{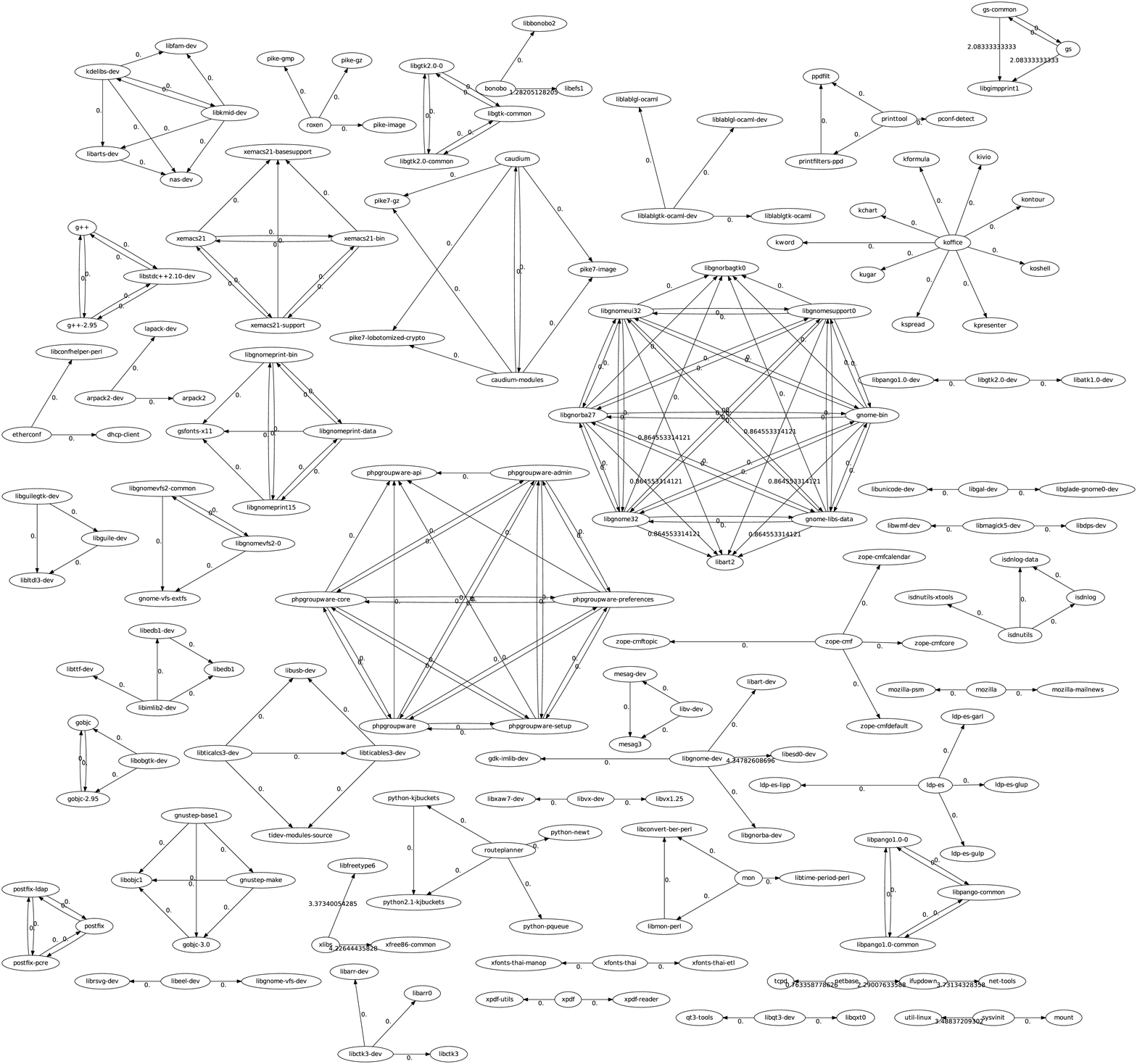}}%

  \caption{Strong dominance graphs for Debian 2.2, 3.0}%
  \label{fig:debian-strong-dom-series2}%
\end{figure}

\begin{figure}
  \center
  \includegraphics[angle=90,width=0.98\textwidth,height=0.98\textheight]{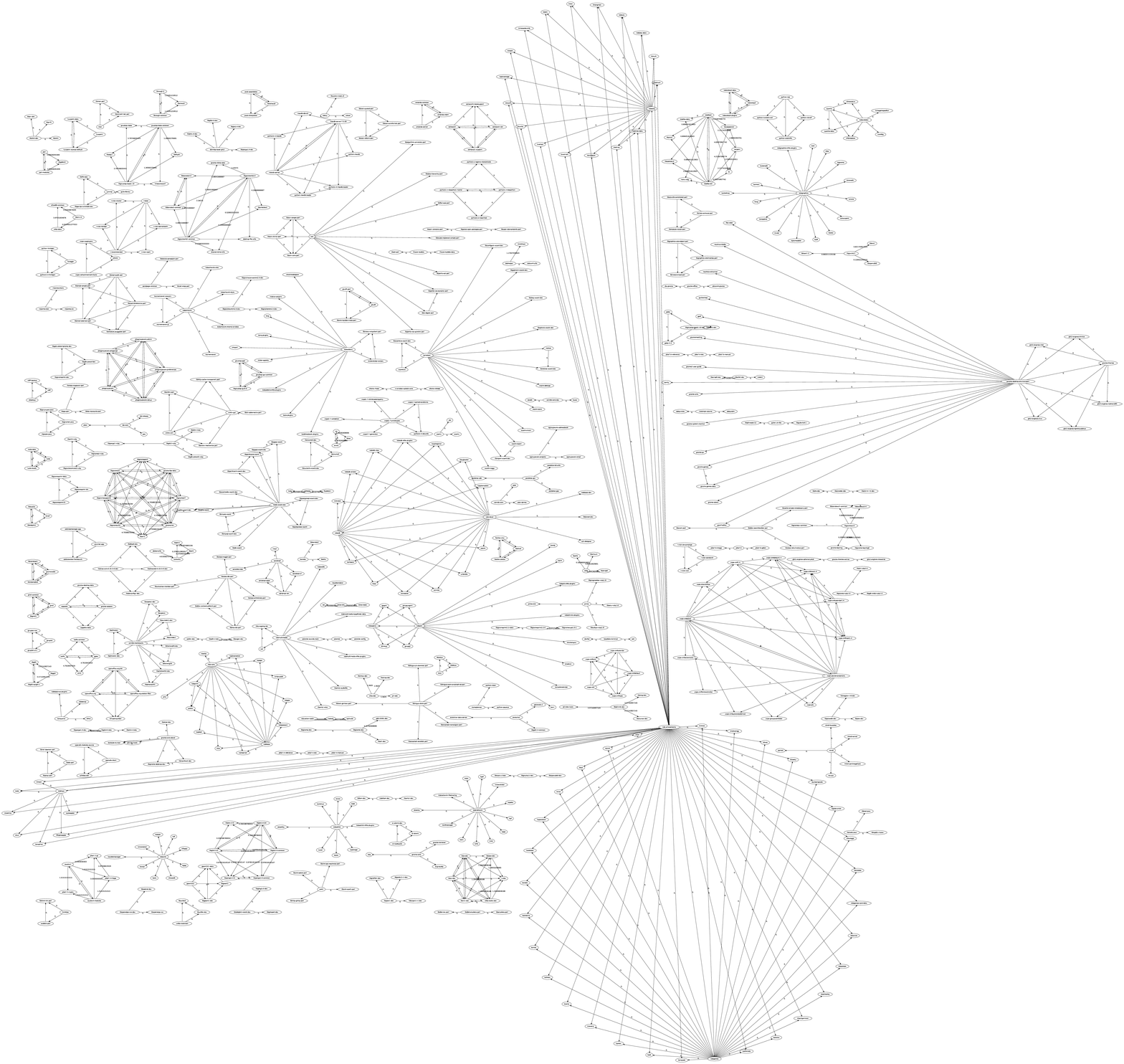}
  \caption{Strong dominance graph for Debian 3.1.\strut}
  \label{fig:debian-strong-dom-series3}
\end{figure}

\begin{figure}
  \center
  \includegraphics[angle=90,width=0.98\textwidth,height=0.98\textheight]{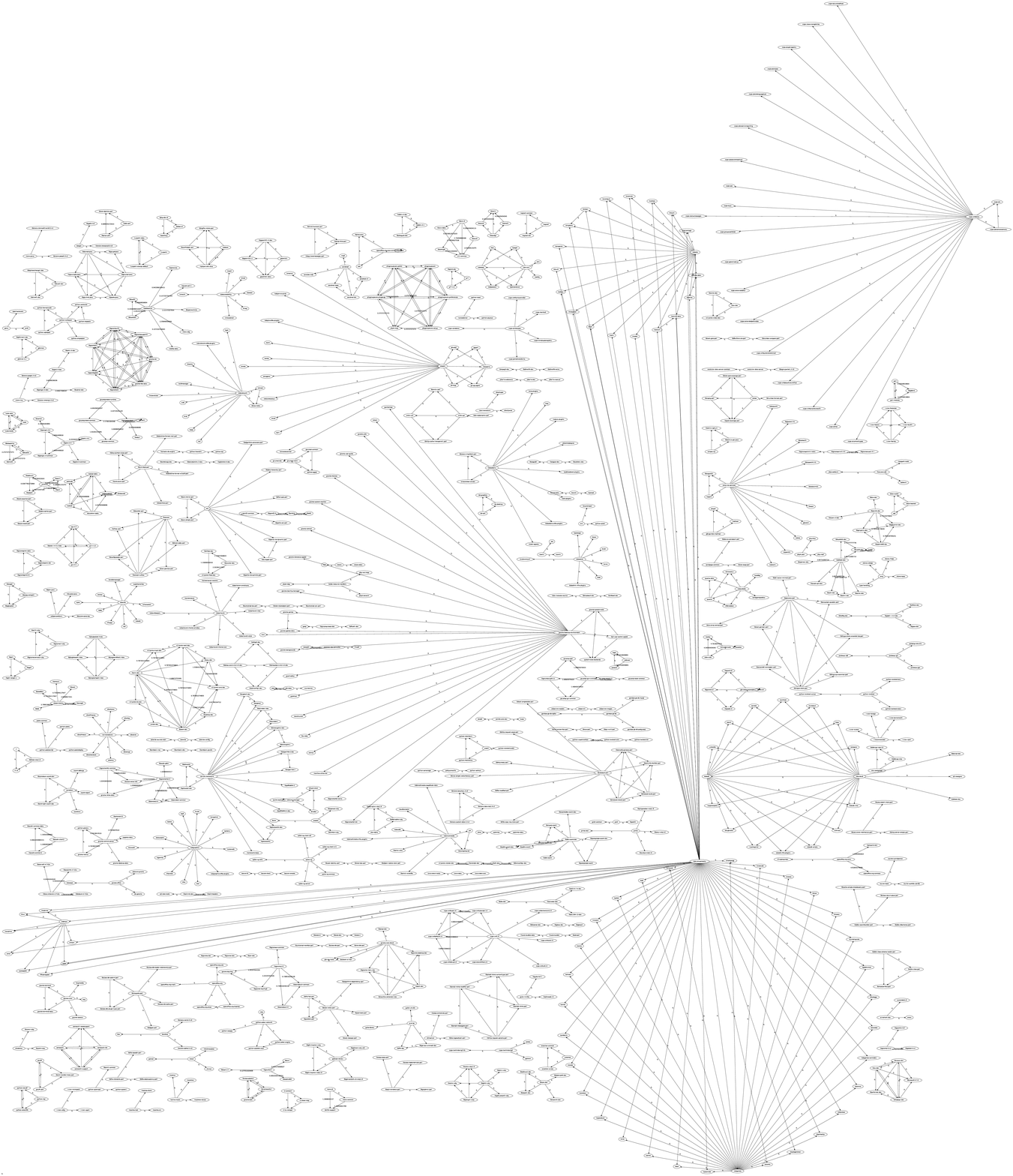}
  \caption{Strong dominance graph for Debian 4.0.\strut}
  \label{fig:debian-strong-dom-series4}
\end{figure}

\begin{figure}[f]
  \center
  \includegraphics[angle=90,width=0.90\textwidth,height=0.90\textheight]{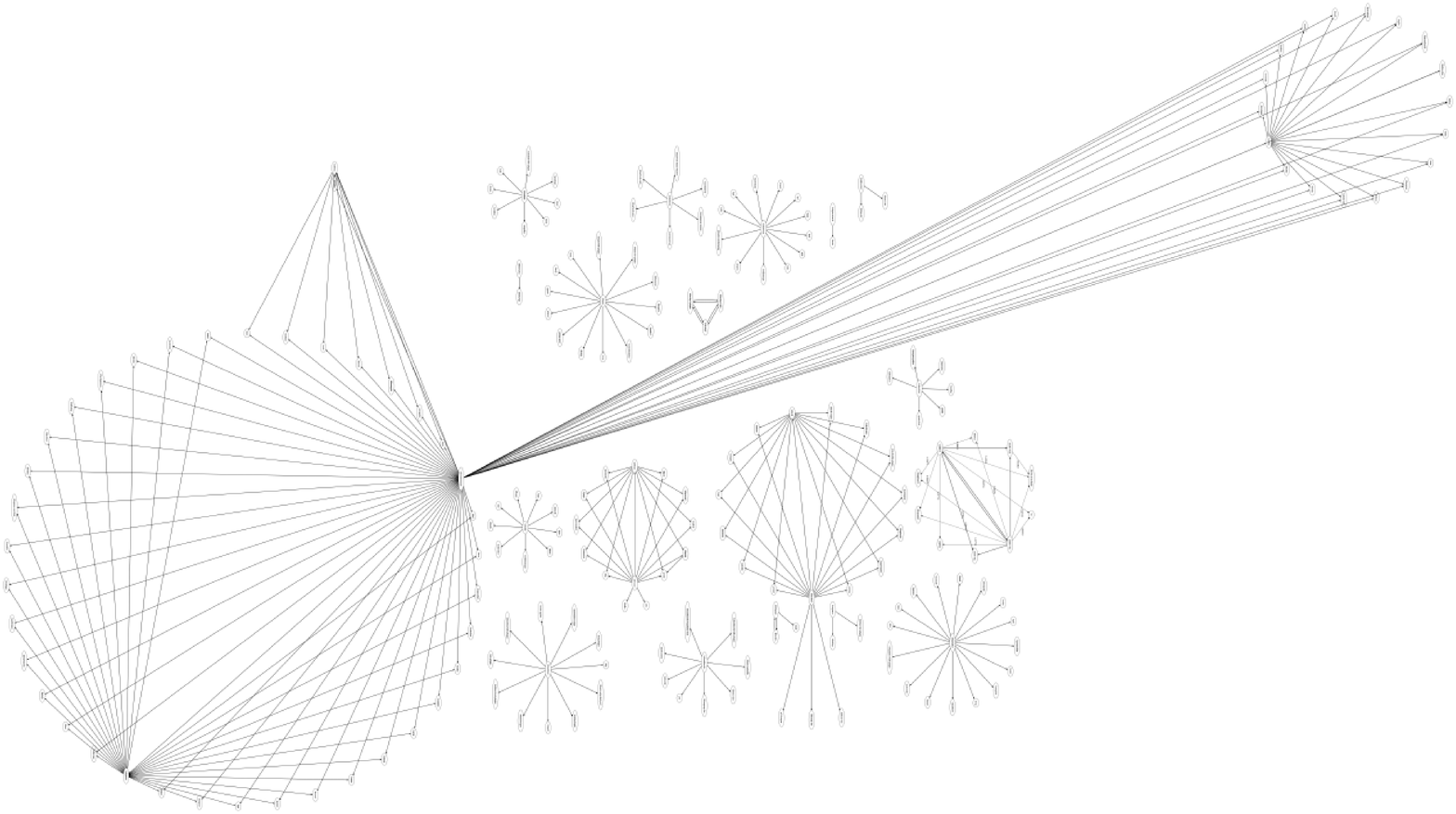}
  \caption{The Figure shows the cluster of packages associated to the
  KDE suite in the strong dominator graph for Debian Sarge.  We can
  immediately spot the package \texttt{kde-amusements}, that is
  partially overlapping \texttt{kde-games} and \texttt{kdeedu}: this
  is an architectural issue, that has been fixed in later Debian
  versions.\strut}
  \label{fig:sarge-kde}
\end{figure}

\begin{figure}
  \center
  \includegraphics[angle=90,width=0.95\textwidth,height=0.95\textheight]{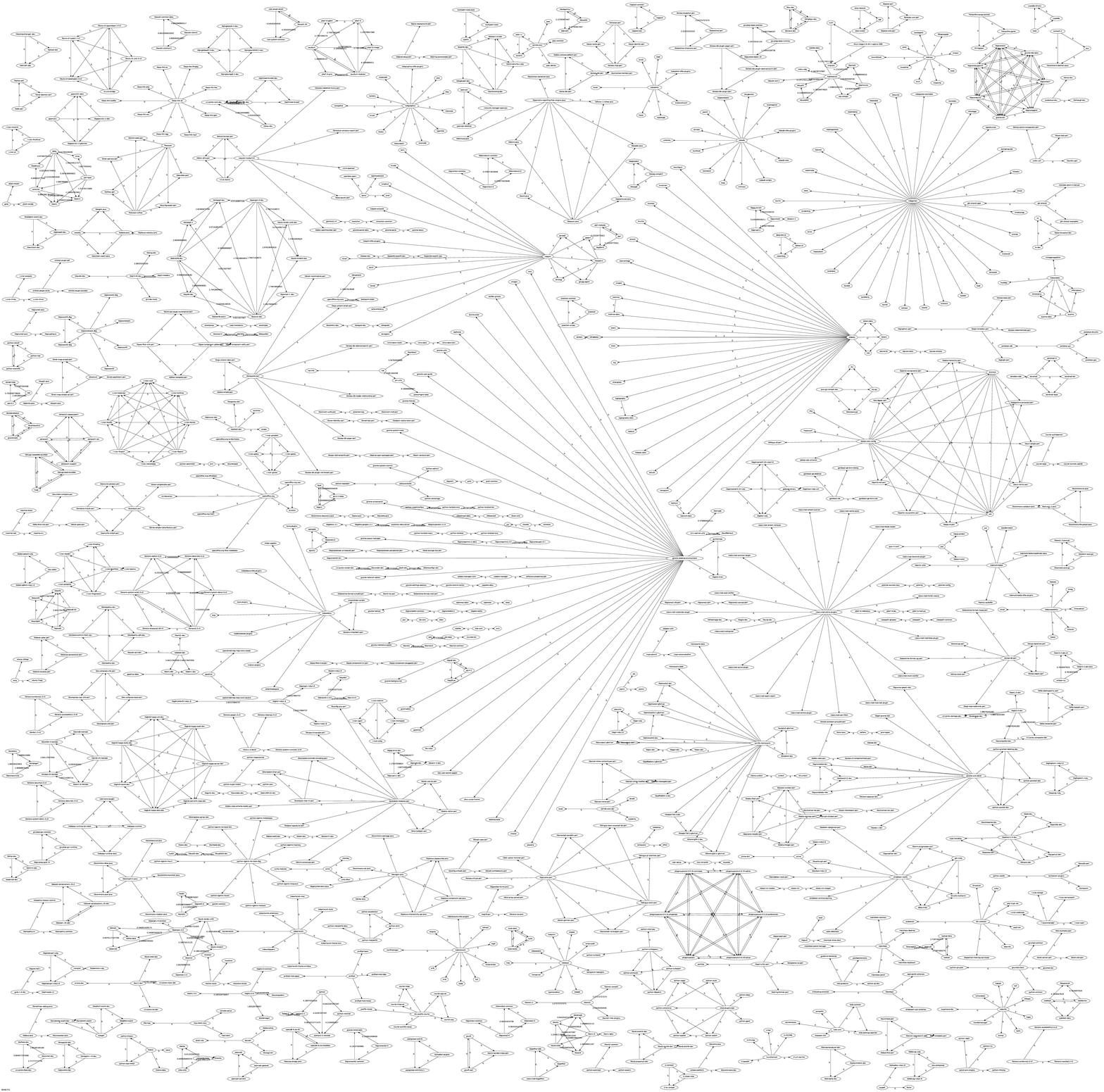}
  \caption{Strong dominance graph for Debian 5.0. The two bigger
  clusters are related to packages coming from the GNOME and KDE
  desktop environments.\strut}
  \label{fig:debian-strong-dom-series5}
\end{figure}

\end{document}